%% file: paper.tex
\documentclass[submission,copyright,creativecommons
]{eptcs}
\providecommand{\event}{The EPTCS post-proceedings of TERMGRAPH 2014} 

\usepackage{amsmath}
\usepackage{amssymb}
\usepackage{amsthm}
\usepackage{amsfonts} 
\usepackage{latexsym}
\usepackage{hyperref}
\usepackage{stmaryrd}

\usepackage[all]{xy}

\usepackage{termgraph2014}

\title{Complexity Analysis of Precedence Terminating Infinite Graph
Rewrite Systems}
\author{Naohi Eguchi%
\thanks{The author is supported by Grants-in-Aid for JSPS Fellows (Grant No.
$25 \cdot 726$).}
\institute{Faculty of Science,
Chiba University,
Japan%
}
\email{neguchi@g.math.s.chiba-u.ac.jp}
}
\def\titlerunning{Complexity Analysis of Precedence Terminating Infinite Graph
Rewrite Systems}
\def\authorrunning{Naohi Eguchi}

\theoremstyle{plain}
\newtheorem{theorem}{Theorem}
\newtheorem{corollary}{Corollary}
\newtheorem{lemma}{Lemma}
\newtheorem{claim}{Claim}

\theoremstyle{definition}
\newtheorem{definition}{Definition}

\theoremstyle{remark}
\newtheorem{remark}{Remark}
\newtheorem{example}{Example}

\begin{document}

\maketitle

\begin{abstract}
The general form of {\em safe recursion} (or {\em ramified recurrence})
 can be expressed by an infinite graph rewrite system including
{\em unfolding graph rewrite rules} introduced by Dal Lago, Martini and
Zorzi, in which the size of every normal
 form by innermost rewriting is polynomially bounded.
Every unfolding graph rewrite rule is {\em precedence terminating} in
 the sense of Middeldorp, Ohsaki and Zantema.
Although precedence terminating infinite rewrite systems cover all the
 primitive recursive functions, in this paper we consider graph rewrite
 systems {\em precedence terminating with argument separation}, which
 form a subclass of precedence terminating graph rewrite systems.
We show that for any precedence
 terminating infinite graph rewrite system $\GS$ with a specific argument separation,  both the runtime complexity of
$\GS$ and the size of every normal form in $\GS$ can be polynomially bounded.
As a corollary, we obtain an alternative proof of the original result
 by Dal Lago et al.
\end{abstract}

\input{intro}
\input{preliminary}
\input{sruf}
\input{spt}
\input{application}

\section{Conclusion}

Generalizing unfolding graph rewrite rules that express the schema
\eqref{eq:gsr}, we proposed restrictive
precedence termination orders,
precedence termination with argument separation.
The restrictive notion together with suitable assumptions
yields a new criterion for
the polynomial runtime complexity of infinite GRSs and for the
polynomial-size normal forms in infinite GRSs.
As discussed in the last section, the proposed method can be potentially expanded for 
safe recursion with multiple recursion arguments or
simultaneous general safe recursion, and thus is indeed more flexible
than unfolding graph rules at least in a limited sense.
It should be stressed, however, that it is unclear how to express
infinite instances of those
recursion schemata with infinite graph rewrite rules in a uniform way.




\end{document}

%% file: intro.tex
\section{Introduction}

\subsection{Backgrounds}

In this paper we present a complexity analysis
of a specific kind of infinite graph rewrite systems,
{\em precedence terminating with argument separation}.
The formulation of precedence termination with argument separation stems from a
function-algebraic characterization of the polytime computable functions 
based on the principle known as {\em safe recursion} \cite{BC92} or
{\em tiered recursion} \cite{Leivant95}.
The schema of safe recursion is a syntactic restriction of the standard primitive
recursion based on a specific separation of argument positions of functions
into two kinds. 
Notationally, the separation is indicated by semicolon as 
$f(\sn{x_1, \dots, x_k}{x_{k+1}, \dots, x_{k+l}})$, 
where $x_1, \dots, x_k$ are called {\em normal} arguments while
$x_{k+1}, \dots, x_{k+l}$
are called {\em safe} ones.
The schema of safe recursion formalizes the idea that 
recursive calls are restricted on normal arguments whereas
substitutions of recursion terms are restricted for safe arguments:
$f(0, \vec y; \vec z) = g(\vec y; \vec z)$,
$f(c_i (x), \vec y; \vec z) = 
h_i (x, \vec y; \vec z, f(x, \vec y; \vec z))$
$(i \in I)$,
where $I$ is a finite set of indices.
Safe recursion is sound for polynomial
runtime complexity over unary constructors, i.e., over numerals or lists,
but it was not clear whether the general form of safe recursion over
arbitrary constructors, which is called 
{\em general ramified recurrence} \cite{GRR2010} or 
{\em general safe recursion}, could be related to polytime computability as well.
\begin{equation}
\label{eq:gsr}
f(c_i (x_1, \dots, x_{\arity (c_i)}), \vec y; \vec z) = 
h_i (\vec x, \vec y; \vec z, f(x_1, \vec y; \vec z), \dots,
                             f(x_{\arity (c_i)}, \vec y; \vec z)
    )
\ (i \in I)
\tag{\textbf{General Safe Recursion}}
\end{equation}
The authors of \cite{GRR2010} answered this question positively
(Theorem \ref{t:GRR10}, Section \ref{s:uf}) showing
that the schema \eqref{eq:gsr} can be expressed by an infinite set of
{\em unfolding graph rewrite rules}.
To see a reason why graph rewriting was employed, consider a term
rewrite system
 $\RS$ over the constructors 
$\{ \epsilon, \m{c}, \m{0}, \ms \}$
consisting of the following four rules with
the argument separation indicated in the rules.
\\ $\mbox{}$ \hfill
%
$
\begin{array}{rclrcl}
\mg (\sn{\epsilon}{z}) & \rightarrow & z & \qquad
\mg (\sn{\m{c} (\sn{}{x, y})}{z}) & \rightarrow &
\m{c} (\sn{}{\mg (\sn{x}{z}), \mg (\sn{y}{z})}) \\
\mf (\sn{\m{0}, y}{}) & \rightarrow & \epsilon &
\mf (\sn{\ms (\sn{}{x}), y}) & \rightarrow & \mg (\sn{y}{\mf (\sn{x, y}{})})
\end{array}
$
\hfill \ \\
Reduction of a term
$\mf (\ms^m (\m{0}), t)$ generates a tree consisting of exponentially
many copies of the tree $t$ measured by $m$.
Thus the computation should be performed over suitably shared graphs rather
than terms.
Moreover, the term
$\mf (\ms (\m{0}), \m{c} (\epsilon, \epsilon))$
leads to the term
$\m{c} (\mg (\epsilon, \epsilon), \mg (\epsilon, \epsilon))$
in three steps, where the subterm
$\mg (\epsilon, \epsilon)$ is duplicated, which means that costly
recomputations potentially occur.
Such duplications cannot be avoided by simple sharing but some essential
memoization technique is necessary.

\subsection{Outline}

The most effort in \cite{GRR2010} was devoted to show that unfolding graph rewrite
rules expressing the schema \eqref{eq:gsr} only yield polynomial lengths of
rewriting sequences and normal forms of polynomial sizes measured by the
sizes of starting (term) graphs.
The initial motivation of the present work was to deduce the complexity result by
means of existing term rewriting techniques.
In a technical report \cite{Eguchi_wst14}, rewriting sequences under
unfolding graph rewrite rules are embedded into descending sequences
under a termination order over lists of terms via a variant of the
{\em predicative interpretation}
\cite{AM05,AEM11,AEM12}.
In this paper, making the definition of unfolding graph
rewrite rules more abstract, we define a class of graph rewrite systems
precedence terminating with argument separation.
Though the complexity analysis in the report above could be adopted, 
we avoid the use of intermediate termination orders but make use of 
numerical interpretation methods, which have been established as well as
termination orders, e.g. \cite{Bonfante01}.
The performed numerical interpretation is closely related to the predicative
interpretation but also strongly motivated by polynomial
{\em quasi-}interpretations presented in 
\cite{Bonfante01LPO,Marion03,Bonfante11}.
After preliminary sections, in Section \ref{s:ptas},
we show that every graph rewrite system precedence terminating with a
specific argument separation reduces under the associated interpretation
(Theorem \ref{t:context}), yielding an alternative proof of 
Theorem \ref{t:GRR10}
(Corollary \ref{c:main}).
In Section \ref{s:app}, to convince readers that the proposed method is
indeed (potentially) more flexible than unfolding graph rewrite rules,
we discuss two possibilities of application referring to related works.

%% file: preliminary.tex
\section{Term graph rewriting}
\label{s:pre}

In this section, we present basics of term graph rewriting mainly
following \cite{BareEGKPS87}.
\begin{definition}[Signatures, labeled graphs and paths]
\normalfont
Let $\FS$ be a {\em signature}, a set of function symbols, and
let 
$\arity : \FS \rightarrow \mathbb N$
where $\arity (f)$ is called the
{\em arity} of $f$.
Throughout of the paper, we only consider finite signatures.
We assume that $\FS$ is partitioned into the set $\CS$ of
constructors and the set $\DS$ of defined symbols.

Let $G = (V_G, E_G)$ be a directed graph consisting of a set
$V_G$ of vertices (or nodes) and a set $E_G$ of directed edges.
A {\em labeled graph} is a triple 
$(G, \lab_G, \att_G)$
of an acyclic directed graph $G = (V_G, E_G)$,
a partial {\em labeling} function
$\lab_G: V_G \rightarrow \FS$ 
and a (total) {\em successor} function
$\att_G: V_G \rightarrow V_G^\ast$, 
mapping a node $v \in V_G$ to a sequence of nodes of length
$\arity (\lab_G)$,
such that
if
$\att_G (v) = v_1, \dots, v_k$, 
then
$\{ v_1, \dots, v_k \} = 
 \{ u \in V_G \mid (v, u) \in E_G \}$.
In case $\att_G (v) = v_1, \dots, v_k$,
the node $v_j$ is called the {\em $j^{\text{th}}$ successor of $v$} for every
$j \in \{ 1, \dots, k \}$.
In particular, $\att_G (v)$ is empty if $\lab_G (v)$ is not defined.

A list 
$\nseq{v_1, m_1, \dots, v_{k-1}, m_{k-1}, v_k}$
consisting of nodes $v_1, \dots, v_m$ of a term graph $G$ and naturals
$m_1, \dots$, $m_{k-1}$ is called a {\em path} from $v_1$ to $v_k$
of {\em length} $k$ if $v_{j+1}$ is the $m_j^{\text{th}}$ successor of $v_j$
for each $j \in \{ 1, \dots, k-1 \}$.
In case $k=0$, the list $\nseq{v}$ consisting of a single node $v$ is a
trivial path of length $0$.
A labeled graph $(G, \lab_G, \att_G)$ is {\em closed} if 
the labeling function $\lab_G$ is total.
\end{definition}
\begin{definition}[Term graphs, sub-term graphs, basic term graphs,
 depths of term graphs and maximal sharing]
A quadruple $(G, \lab_G, \att_G, \rootnode_G)$ is a {\em term graph} if 
$(G, \lab_G, \att_G)$ is a labeled graph and
$\rootnode_G$ is the {\em root} of $G$, i.e.,
a unique node in $V_G$ from which every node is reachable.
We write $\mathcal{TG(F)}$ to denote the set of term graphs over a
 signature $\mathcal F$.
For a labeled graph 
$G = (G, \att_G, \lab_G)$
and a node $v \in V_G$,
$G \seg v$ denotes the {\em sub-term graph} of $G$ rooted at $v$.
We write $H \subg G$ to express that $H$ is a sub-term graph of $G$ and 
$\ssubg$ for the proper relation.
A term graph $G \in \TG{F}$ is called {\em basic} if
$\lab_G (\rootnode_G) \in \DS$
and
$G \seg v \in \TG{C}$ 
for every successor node $v$ of $\rootnode_G$.
For a term graph $G$, the length of the longest path(s) from
$\rootnode_G$ the {\em depth} of $G$, denoted as $\depth (G)$.

Undefined nodes in a term graph $G$ are intended to behave as free
variables in a natural term representation of $G$.
Let 
$\term_G$ be an injective mapping from undefined nodes in $G$
to a (possibly infinite) set
$\VS$ of variables.
The mapping $\term_G$ is canonically extended to a term
representation (over $\FS \cup \VS$) of sub-term graphs of
$G$ as
$\term_G (G \seg v) = \term_G (v) \in \VS$
in case $\lab_G (v)$ is not defined, and otherwise
$\term_G (G \seg v) = \lab_G (v)
 (\term_G (G \seg v_1), \dots, \term_G (G \seg v_k))$
where $\att_G (v) = v_1, \dots, v_k$.
A term graph $G$ is
{\em maximally shared} if,
for any two nodes $u, v \in V_G$,
$\term_G (G \seg u) = \term_G (G \seg v)$ implies
$u = v$
(under an arbitrary choice of such a mapping $\term_G$).
\end{definition}
\begin{definition}[Homomorphisms, redexes, graph rewrite rules and
 constructor graph rewrite rules]
Given two labeled graphs $G$ and $H$,
a {\em homomorphism} from $G$ to $H$ is a mapping
$\varphi: V_G \rightarrow V_H$ 
such that 
$\lab_H (\varphi (v)) = \lab_G (v)$ for each 
$v \in \dom{\lab_G} \subseteq V_G$ and that,
for each $v \in \dom{\lab_G}$, if $\att_G (v) = v_1, \dots, v_k$, then
$\att_H (\varphi (v)) = \varphi (v_1), \dots, \varphi (v_k)$.
By definition, these conditions are not required for a node $v \in V_G$ for which
$\lab_G (v)$ is not defined.
A homomorphism $\varphi$ from a term graph $G$ to another term graph
$H$ is a homomorphism 
$\varphi : (G, \lab_G, \att_G) \rightarrow (H, \lab_H, \att_H)$
such that $\rootnode_H = \varphi (\rootnode_G)$.

A {\em graph rewrite rule} is a triple $\rho = (G, l, r)$ of a labeled
graph $G$ and two distinct nodes $l$ and $r$ respectively called
the {\em left} and {\em right} root. 
The term rewrite rule
$\mg (x, y) \rightarrow \m{c} (y, y)$ 
is expressed by a graph rewrite rule (1) and
$\mh (x, y, z, w) \rightarrow \m{c} (z, w)$
by (2) below.
\\
$\mbox{}$ \hfill
$
  (1) \quad
  \xymatrix{*+[o][F-]{\mg} \ar[d] \ar[dr] &
            *+[F]{\m{c}} \ar@/_/[d] \ar@/^/[d] \\
             \bot & \bot
           }
  \qquad \qquad
  (2)
  \xymatrix{& *+[o][F-]{\mh} \ar[dl] \ar[d] \ar[dr] \ar[drr] & &
            *+[F]{\m{c}} \ar[dl] \ar[d] \\
            \bot & \bot & \bot & \bot
           }
$
\hfill $\mbox{}$ \\
In the examples, the left root is written in a circle while the right
root is in a square, and
undefined nodes are indicated as $\bot$.
As in the usual term rewriting setting, we assume that every undefined
node occurring in $G \seg r$ occurs in $G \seg l$ as well.

A {\em redex} in a term graph $G$ is a pair 
$(R, \varphi)$ of a rewrite rule
$R = (H, l, r)$ and a homomorphism
$\varphi : H \seg l \rightarrow G$.
A set $\GS$ of graph rewrite rules is called a 
{\em graph rewrite system} (GRS for short).
A graph rewrite rule $(G, l, r)$ is called a {\em constructor} one if
$G \seg l$ is a basic term graph.
A GRS $\GS$ is called a constructor one if $\GS$ consists only of
constructor rewrite rules.
The rewrite relation in a GRS $\GS$ is defined by the {\em build},
{\em redirection} and {\em garbage collection} phase,
 denoted as $\rightarrow_{\GS}$
(See, e.g.,  \cite{GRR2010}).
In case that  $G \rew H$ is induced by a redex $((K, l, r), \varphi)$,
one can find a homomorphism
$\psi :K \seg r \rightarrow H$ compatible with $\varphi$ such that
$G$ results in $H$ by replacing the sub-term graph
$G \seg \varphi (l)$ of $G$ with 
$H \seg \psi (r)$.
A formal definition can be found in \cite{BareEGKPS87}.
The $m$-fold iteration of $\rew$ is denoted as 
$\rewm{m}$ and the reflective and transitive closure as $\rewast$.
A rewriting $G \rew H$ induced by a redex $((K_0, l_0, r_0), \varphi_0)$ is
{\em innermost} if there is no redex $((K, l, r), \varphi)$ such that
$G \seg \varphi (l)$ is a proper sub-term graph of
$G \seg \varphi_0 (l_0)$.
The innermost rewrite relation in $\GS$ is denoted as 
$\irew$, and $\irewm{m}$, $\irewast$ are defined
accordingly.
\end{definition}

%% file: sruf.tex
\section{Unfolding graph rewrite rules for general safe recursion}
\label{s:uf}

To make the purpose of the present work precise, in this section  we
restate the main result in \cite{GRR2010},
formulating the {\em general safe recursive} functions.
Let $\CS$ be a set of constructors and 
$m \mapsto c_m$ ($1 \leq m \leq |\CS|$) be an enumeration for $\CS$.
We assume that $\CS$ contains at least one constant.
We call a function 
$f: \mathcal{T(C)}^{k+l} \rightarrow \mathcal{T(C)}$
{\em general safe recursive} if, under a suitable argument separation
$f(\sn{x_1, \dots, x_k}{y_1, \dots, y_l})$,
$f$ can be defined from the initial functions by operating the schemata
specified as follows.
\begin{itemize}
\item $O^{k, l}_j (\sn{x_1, \dots, x_k}{y_1, \dots, y_l}) = c_j$
      if $c_j$ is a constant.
\hfill {\bf (Constants)}
\item $C_j (; x_1, \dots, x_{\arity (c_j)}) = 
       c_j (x_1, \dots, x_{\arity (c_j)})$
      if $\arity (c_j) \neq 0$.
\hfill {\bf (Constructors)}
\item $I^{k, l}_j (\sn{x_1, \dots, x_k}{x_{k+1}, \dots, x_{k+l}}) = 
       x_{j}
      $
      $(1 \leq j \leq k+l)$.
\hfill {\bf (Projections)}
\item $P_{i, 0} (\sn{}{c_i}) = c_i$ $(\arity (c_i) = 0)$,
      $P_{i, j} (\sn{}{c_i (x_1, \dots, x_{\arity (c_i)})
                      }
                )
       = x_j$ $(1 \leq j \leq \arity (c_i))$.
\hfill {\bf (Predecessors)}
\item $C(\sn{}%
            {c_{j} (x_1, \dots, x_{\arity (c_j)}), 
             y_1, \dots, y_{|\CS|}
            }) =
       y_{j}$.
\hfill {\bf (Conditional)}
\item $f(\sn{x_1, \dots, x_k}{y_1, \dots, y_l}) =
       h(\sn{x_{j_1}, \dots, x_{j_m}}%
            {g_1 (\sn{\vec x}{\vec y}), \dots,
             g_n (\sn{\vec x}{\vec y})
            }
        )$
$(\{ j_1, \dots, j_m \} \subseteq \{ 1, \dots, k \})$,
\\
where $h$ has $m$ normal and $n$ safe argument positions.
\hfill {\bf (Safe composition)}
\item $f(\sn{c_j (x_1, \dots, x_{\arity (c_j)}), \vec y}{\vec z}) =
       h_j (\sn{\vec x, \vec y}%
               {\vec z, f(\sn{x_1, \vec y}{\vec z}), \dots,
                        f(\sn{x_{\arity (c_j)}, \vec y}{\vec z})
               }
           )$
$(j \in I)$.
\\
If $c_j$ is a constant, the schema of denotes
$f(\sn{c_j, \vec y}{\vec z}) = h_j (\sn{\vec y}{\vec z})$.
\hfill {\bf (General safe recursion)}
\end{itemize}

In \cite{GRR2010} a GRS $\GS$ is called 
{\em polytime presentable} if there exists a deterministic polytime
algorithm which, given a term graph $G$, returns a term graph $H$ such
that $G \irew H$ if such a term graph exists, or
the value $\m{false}$ if otherwise.
In addition, a GRS $\GS$ is {\em polynomially bounded} if there exists a
polynomial $p: \mathbb N \rightarrow \mathbb N$ such that
$\max \{ m, |H| \} \leq p (|G|)$ holds whenever $G \irewm{m} H$ holds.

\begin{theorem}[Dal Lago, Martini and Zorzi \cite{GRR2010}]
\label{t:GRR10}
Every general safe recursive function can be computed by a
polytime presentable and polynomially bounded constructor GRS. 
\end{theorem}

\begin{remark}
The schema \eqref{eq:gsr} is formulated based on safe recursion (on
 notation) following \cite{BC92} whereas
the schema of general ramified recurrence formulated in \cite{GRR2010}
 is based on ramified recurrence following \cite{Leivant95}.
Due to the difference, the
 definition of general safe recursive functions above
 is slightly different from the original definition of
{\em tiered recursive} functions in \cite{GRR2010}.
Notably, the schema (\textbf{Safe composition}) is a weaker form of the
 original one in \cite{BC92}, which was introduced in \cite{HW99}.
It is not clear whether there is a precise correspondence between
 general safe recursive functions in the current formulation and tiered
 recursive functions.
However, it is known that the polytime functions (over binary words) can
 be covered with the weak form of safe composition,
which means that the restriction of the
 general safe recursive functions to unary constructors still
covers all the polytime functions.
\end{remark}

Theorem \ref{t:GRR10} is shown by induction over a general safe
recursive function $f$.
The case that $f$ is defined by \eqref{eq:gsr} is witnessed by an infinite set of 
{\em unfolding graph rewrite rules}.

\begin{definition}[Unfolding graph rewrite rules]
\label{d:uf}
\normalfont
Let $\Sigma$ and $\Theta$ be two disjoint signatures in a bijective
 correspondence by 
$\varphi: \Sigma \rightarrow \Theta$.
For a fixed $k \in \mathbb N$, suppose that
$\arity ( \varphi (g)) = 2 \arity (g) + k$
for each $g \in \Sigma$.
Let
$f \not\in \Sigma \cup \Theta$ and
$\arity (f) = 1+k$.
Given a natural $m \geq 1$, the
{\em $m^{\text{th}}$ set of unfolding graph rewrite rule over
$\Sigma$ and $\Theta$ defining $f$}
consists of graph rewrite rules of the form
$(G, l, r)$ where
$G = (V_G, E_G, \att_G, \lab_G)$
is a labeled graph over a signature 
$\FS \supseteq \Sigma \cup \Theta$
that fulfills the following conditions.
\begin{enumerate}
\item The set $V_G$ of vertices consists of $1 + 2m + k$ elements
      $y$, $v_1, \dots, v_m$, $w_1, \dots$, $w_m$, $x_1, \dots, x_k$.
\item $l = y$ and $r = w_1$.
\item $\lab_G (y) = f$ and
      $\att_G (y) = v_1, x_1, \dots, x_k$.
\label{d:uf:f}
\item $\lab_G (x_j)$ is not defined for all
      $j \in \{ 1, \dots, k \}$.
\item For each $j \in \{ 1, \dots, m \}$,
      $\att_G (v_j) \in \{ v_1, \dots, v_m \}^\ast$.
      Moreover,
      $V_{G \Seg v_1} = \{ v_1, \dots, v_m \}$.
\label{d:uf:v}
\item For each $j \in \{ 1, \dots, m \}$,
      $\lab_G (v_j) \in \Sigma$ and
      $\lab_G (w_j) = \varphi (\lab_G (v_j))$.
\item For each $j \in \{ 1, \dots, m \}$,
      $\att_G (w_j) = 
      v_{j_1}, \dots, v_{j_n}, x_1, \dots, x_k, w_{j_1}, \dots, w_{j_n}$
      if
      $\att_G (v_j) = v_{j_1}, \dots, v_{j_n}$.
\label{d:uf:w}
\end{enumerate}
\end{definition}
\begin{example}
\label{ex:ugr}
Let $\Sigma = \{ \m{0}, \ms \}$,
$\Theta = \{ \mg, \mh \}$,
$\varphi : \Sigma \rightarrow \Theta$
be a bijection defined as
$\m{0} \mapsto \mg$ and $\ms \mapsto \mh$,
and
$\mf \not\in \Sigma \cup \Theta$, where
the arities of $\m{0}, \ms, \mg, \mh, \mf$
are respectively $0$, $1$, $1$, $3$ and $2$.
Namely we consider the case $k=1$.
The standard equations 
$\mf (\m{0}, x) \rightarrow \mg (x)$,
$\mf (\ms (y), x) \rightarrow \mh (y, x, \mf (y, x))$
of primitive recursion can be expressed by the infinite GRS
$\bigcup_{m \geq 1} \GS_m$,
where $\GS_m$ is the $m^{\text{th}}$ set of unfolding graph rewrite rules over 
$\FS = \Sigma \cup \Theta \cup \{ \mf \}$ defining $\mf$.
In this case, for each $m \geq 1$, $\GS_m$ consists of a
 single rule.
For example, in case
$i = 1, 2, 3$, $\GS_i$ consists of the rewrite rule $(i)$
pictured in Figure 
\ref{fig:ugr}.
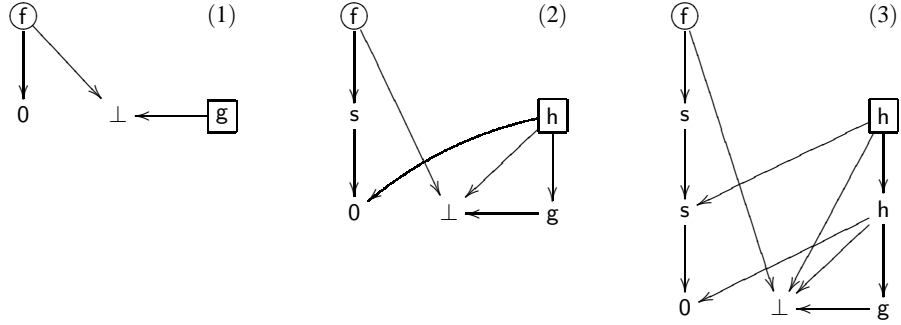
\begin{figure}[t]
\footnotesize
\begin{equation*}
  \xymatrix{*+[o][F-]{\mf} \ar[d] \ar[dr] & & (1) \\
            \m{0} & \bot & *+[F]{\mg} \ar[l]
           }
  \qquad \qquad
  \xymatrix{*+[o][F-]{\mf} \ar[d] \ar[ddr] & & (2) \\
            \ms \ar[d] & & *+[F]{\mh} \ar@/_/[dll] \ar[dl] \ar[d] \\
            \m{0} & \bot & \mg \ar[l]
           }
  \qquad \qquad
  \xymatrix{*+[o][F-]{\mf} \ar[d] \ar[dddr] & & (3) \\
            \ms \ar[d] & &
            *+[F]{\mh} \ar[dll] \ar[ddl] \ar[d] \\
            \ms \ar[d] & & \mh \ar[dll] \ar[dl] \ar[d] \\
            \m{0} & \bot & \mg \ar[l]
           }
\end{equation*}
\caption{Examples of unfolding graph rewrite rules}
\label{fig:ugr}
\end{figure}
As seen from the pictures, the unfolding graph rewrite rules in Figure \ref{fig:ugr} express the infinite instances
$\mf (\m{0}, x) \rightarrow \mg (x)$,
$\mf (\ms (\m{0}), x) \rightarrow \mh (\m{0}, x, \mg (x))$,
$\mf (\ms (\ms (\m{0})), x) \rightarrow 
 \mh (\ms (\m{0}), x, \mh (\m{0}, x, \mg (x)))$, ...,
representing terms as suitably shared term graphs.
\end{example}

To keep every term graph compatible with the associated argument separation, in \cite{GRR2010}, for
any redex $(R, \varphi)$, the homomorphism $\varphi$ is limited to
an {\em injective} one.
In this paper, instead of assuming injectivity of homomorphisms, we rather indicate
argument separations explicitly.

\begin{definition}[Term graphs with argument separation]
\label{d:safeterm}
\normalfont
In accordance with the idea of safe recursion, we assume that the argument
 positions of every function symbol are separated into the normal
 and safe ones, writing 
$f(\sn{x_1, \dots, x_k}{x_{k+1}, \dots, x_{k+l}})$
to denote $k$ normal arguments and $l$ safe ones.
We always assume that every constructor symbol in $\CS$ has safe
argument positions only.
The argument separations of function symbols are taken into account in labeled graphs in such a way that
for every successor $u$ of a node $v$ 
we write $u \in \nrm (v)$ if $u$ is connected to a normal argument
 position of $\lab_G (v)$, and 
$u \in \safe (v)$ if otherwise.
For two distinct nodes $v_0$ and $v_1$, if
$\lab_G (v_0) = \lab_G (v_1)$, then,
for any $j \in \{ 1, \dots, \arity (\lab_G (v_0)) \}$,
$u_0 \in \nrm (v_0) \Leftrightarrow u_1 \in \nrm (v_1)$
for the $j^{\text{th}}$ successor $u_i$ of $v_i$ ($i = 0, 1$).
Notationally, we write 
$\att_G (v) =
 \sn{v_1, \dots, v_k}{v_{k+1}, \dots, v_{k+l}}$
to express the separation such that
$v_1, \dots, v_k \in \nrm (v)$
and
$v_{k+1}, \dots, v_{k+l} \in \safe (v)$.
We assume that any homomorphism
$\varphi : G \rightarrow H$
preserves argument separations.
Namely, for each $v \in \dom{\lab_G}$, if 
$\att_G (v) = v_1, \dots, v_k$; 
$v_{k+1}, \dots, v_{k+l}$, then
$\att_H (\varphi (v)) = 
 \varphi (v_1), \dots, \varphi (v_k)$;
$\varphi (v_{k+1}), \dots, \varphi (v_{k+l})$.
\end{definition}

Let us recall the idea of safe recursion that the number of recursive
calls is measured only by a normal argument and recursion terms can be
substituted only for safe arguments.
This motivates us to introduce the following
safe version of unfolding graph rewrite rules.

\begin{definition}[Safe recursive unfolding graph rewrite rules]
\label{d:sruf}
\normalfont
We call an unfolding graph rewrite rule {\em safe recursive} if the
 following constraints imposed on the clause \ref{d:uf:f} and
 \ref{d:uf:w} in Definition \ref{d:uf} are satisfied.
\begin{enumerate}
\item In the clause \ref{d:uf:f},
      $v_1 \in \nrm (y)$, and
      in the clause \ref{d:uf:w},
      $v_{j_1}, \dots, v_{j_n} \in \nrm (w_j)$ and
      $w_{j_1}, \dots, w_{j_n} \in \safe (w_j)$.
\item In the clause \ref{d:uf:f} and \ref{d:uf:w},
      for each $j \in \{ 1, \dots, k \}$,
      $x_j \in \nrm (y)$ if and only if
      $x_j \in \nrm (w_i)$ for all $i \in \{ 1, \dots, m \}$.
\end{enumerate}
\end{definition}

As a consequence of Definition \ref{d:sruf}, we have a basic property of
safe recursive unfolding graph rewrite rules, which ensures that
rewriting by any unfolding  graph rewrite rule does not change the structures of
subgraphs in normal argument positions.

\begin{corollary}
\label{c:sruf}
Let $(G, y, w_1)$ be a safe recursive unfolding graph rewrite rule with
the set $V_G$ of vertices consisting of $1+2m+k+l$ elements 
$y$, $v_1, \dots, v_m$, $w_1, \dots, w_m$, $x_1, \dots, x_{k+l}$
specified as in Definition \ref{d:uf} and \ref{d:sruf},
where 
$\att_G (y)$ $= \sn{v_1, x_1, \dots, x_k}{x_{k+1}, \dots, x_{k+l}}$.
Then
$G \seg u \ssubnrm G \seg y $
holds for any $j \in \{ 1, \dots, m \}$
and for any node
$u \in \nrm (w_j)$.
\end{corollary}
Corollary \ref{c:sruf} follows from an observation that,
for any $j \in \{ 1, \dots, m \}$
and for any node
$u \in \nrm (w_j)$, either
$u \in \dom{\lab_G}$ and $u = v_i$ for some
$i \in \{ 1, \dots, m \}$,
or
$u \not\in \dom{\lab_G}$ and $u = x_i$ for some
$i \in \{ 1, \dots, k \}$.
This is exemplified by a safe recursive (constructor) unfolding graph
rewrite rule in Figure \ref{fig:srugr} in case
$m=2$ and $l=k=1$ that expresses the term rewrite rule
$\mf (\sn{\ms (\sn{}{\m{0}}), x}{y}) \rightarrow
 \mh (\sn{\m{0}, x}{y, \mg (\sn{x}{y})})$.
To make a contrast, every edge
$\xymatrix{v \ar[r] & u}$ is expressed as
$\xymatrix{v \ar@{-->}[r] & u}$ if
$u \in \safe (v)$ and $\lab_G (v) \in \DS$
and as
$\xymatrix{v \ar@{..>}[r] & u}$ if
$\lab_G (v) \in \CS$.
\begin{figure}[t]
\begin{equation*}
  \xymatrix{*+[o][F-]{\mf} \ar[d] \ar[ddr] \ar@/^/@{-->}[ddrr] & & & \\
            \ms \ar@{..>}[d] & & &
            *+[F]{\mh} \ar@/_/[dlll] \ar@/_/[dll]
                       \ar@{-->}[dl] \ar@{-->}[d] \\
            \m{0} & \bot & \bot &
            \mg \ar@/_/[ll] \ar@{-->}[l]
           }
\end{equation*}
\caption{Example of a safe recursive unfolding graph rewrite rule}
\label{fig:srugr}
\end{figure}
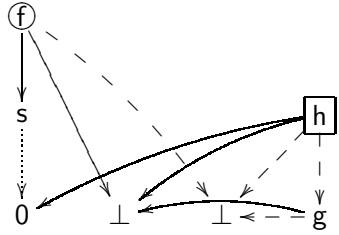

%% file: spt.tex
\section{Precedence termination with argument separation}
\label{s:ptas}

Every unfolding graph rewrite rule is {\em precedence terminating} in
the sense of \cite{MidOZ96}.
In this section we propose a restriction of the standard precedence
termination orders, {\em precedence termination with argument separation}. 
To show the polynomial runtime complexity of those GRSs,
we also introduce a non-standard number-theoretic interpretation of GRSs
precedence terminating with argument separation.

Let $\FS = \CS \cup \DS$ be a signature.
A {\em precedence} $\sp$ is a well founded partial binary relation on
$\FS$.
The {\em rank} $\rk : \FS \rightarrow \mathbb N$
is defined to be compatible with $\sp$:
$\rk (g) < \rk (f) \Leftrightarrow g \sp f$.
We always assume that every constructor symbol is $\sp$-minimal. 

\begin{definition}[A restrictive sub-term graph relation
$\ssubnrm$ and precedence termination with argument separation]
\label{d:ptas}
\normalfont
Let $\sp$ be a precedence on a signature $\FS$ and
$G, H \in \mathcal{TG(F)}$.
We write 
$H \ssubnrm G$
if $H \subg G \seg v$ holds for some node
$v \in \nrm (\rootnode_G)$.
The relation $H \spt G$
holds if 
$\lab_H (v) \sp \lab_G (\rootnode_G)$
for any $v \in V_H$ whenever $\lab_H (v)$ is defined, and additionally
one of the following two cases holds.
\begin{enumerate}
\item $H \eqspt G \seg u$ for some successor node
      $u$ of $\rootnode_G$.
\label{d:ptas:1}
\item $\lab_H (\rootnode_H)$ is defined, i.e.
      $\lab_H (\rootnode_H) \sp \lab_G (\rootnode_G)$,
      \begin{itemize}
      \item $H \ssubnrm G$ for each $v \in \nrm (\rootnode_H)$, and
      \item $H \seg v \spt G$ for each $v \in \safe (\rootnode_H)$.
      \end{itemize}
\label{d:ptas:2}
\end{enumerate} 
We say that a GRS $\GS$ over a signature $\FS$ is
{\em precedence terminating with argument separation}
if for some separation of argument positions and for some precedence
$\sp$ on $\FS$, 
$G \seg r \spt G \seg l$ holds for each rule 
      $(G, l, r) \in \GS$ for the relation 
      $\spt$ induced by the precedence $\sp$. 
\end{definition}

Let us recall we always assume that every constructor is minimal in any
precedence.
By the minimality, for any constructor rewrite rule
$(G, l, r) \in \GS$, if
$G \seg r \spt G \seg l$ holds by Case \ref{d:ptas:1} of Definition
\ref{d:ptas},
$G \seg r \subg G \seg v$
holds for some successor node $v$ of $l$.

\begin{definition}[Safe paths and a class 
$\TGnrm$ of terms]
\label{d:safepath}
\normalfont
\
\begin{enumerate}
\item A path $\nseq{v_1, m_1, \dots, v_{k-1}, m_{k-1}, v_k}$
      in a term graph $G$ is called a {\em safe} one if
      $v_{j+1} \in \safe (v_j)$ for all $j \in \{ 1, \dots, k-1 \}$.
Notationally, for a term graph $G$ and two nodes $u, v \in V_G$,
we write $u \in \safepath_G (v)$ if $u$ lies on a safe path
from $v$ in $G$.
We will also use the notation $\safepath_G (v)$ to denote the set of such nodes $u$.
The relation
$u \in \safepath_G (\rootnode_G)$ will be simply written as
$u \in \safepath_G$.
\label{d:safepath:1}
\item Given a signature $\FS = \CS \cup \DS$,
      we define a subset 
      $\mathcal{TG}_{\nrm} (\FS) \subseteq \mathcal{TG(F)}$.
      It holds $G \in \mathcal{TG}_{\nrm} (\FS)$ if
      $G \in \mathcal{TG(C)}$, or
      $G \seg v \in \mathcal{TG(C)}$ for each
      $v \in \nrm (\rootnode_G)$ and
      $G \seg v \in \mathcal{TG}_{\nrm} (\FS)$ for each
      $v \in \safe (\rootnode_G) $.
\label{d:safepath:2}
\end{enumerate}
\end{definition}
By definition, the root $\rootnode_G$ of $G$ lies on the trivial safe path from 
$\rootnode_G$ in $G$.
In the graph rewrite rule
$(G, l, r)$ in Figure \ref{fig:srugr},
visually every safe path consists only of dashed edges
$\xymatrix{\cdot \ar@{-->}[r] & \cdot}$.
Thus, for non-trivial examples, the right bottom $\bot$ lies on a safe path from
$l$, and both the same $\bot$ and $\mg$ lie on 
a safe path from $r$.
The definition of the class $\TGnrm$ yields
$G \in \mathcal{TG}_{\nrm} (\FS)$
for any basic term graph $G \in \mathcal{TG(F)}$.

\begin{lemma}     
\label{l:safepath}
Let $\GS$ be a constructor GRS over a signature $\FS$ and
$G \in \mathcal{TG}_{\nrm} (\FS)$.
\begin{enumerate}
\item  Let
$(H, l, r) \in \GS$ be a rewrite rule and
$\varphi : H \seg l \rightarrow G$
a homomorphism.
Then any path from $\rootnode_G$ to $\varphi (l)$ is a safe path.
\label{l:safepath:1}
\item Suppose additionally that $\GS$ is precedence
 terminating with argument separation.
If $G \rew H$, then
$H \in \mathcal{TG}_{\nrm} (\FS)$.
\label{l:safepath:2}
\end{enumerate} 
\end{lemma}

\begin{proof}
{\sc Property} \ref{l:safepath:1}.
We show the property by contradiction.
Assume that there exists a path from $\rootnode_G$ to $\varphi (l)$ that
 is not safe.
Then the path passes a normal argument
 position of an intermediate node $v$.
Since constructors have only safe argument positions,
$\lab_G (v)$ must be a defined symbol.
Hence $G \seg \varphi (l) \in \mathcal{TG} (\CS)$
by the definition of the class 
$\mathcal{TG}_{\nrm} (\FS)$.
But
$\lab_G (\varphi (l)) = \lab_{H} (l) \in \DS$
since $\GS$ is a constructor GRS, contradicting
$G \seg \varphi (l) \in \mathcal{TG} (\CS)$.

{\sc Property} \ref{l:safepath:2}.
Suppose that $G$ results in $H$ by applying a redex 
$(R, \varphi)$ for a rule $R = (K, l, r) \in \GS$.
Since any path from $\rootnode_G$ to $\varphi (l)$ is a safe one by Property
 \ref{l:safepath:1}, it suffices to show that 
$H \seg r_H \in \TGnrm$ for the node 
$r_H \in H$ corresponding to $r \in V_{K}$.
We show that $H \seg r_H \in \TGnrm$ holds by induction according to the
 definition of the relation $\spt$.

{\sc Case.} $K \seg r \eqspt K \seg u$
for some successor node $u$ of $l$:
In this case, since $\GS$ is a constructor GRS,
$K \seg r$ is a sub-term graph of $K \seg l$
as noted after Definition \ref{d:ptas}.
Hence
$H \seg r_H \in \TGnrm$
follows from the assumption
$G \seg \varphi (u) \in \TGnrm$.
 
{\sc Case.} Otherwise:
For each $v \in \nrm (r)$, $K \seg r$ is a sub-term graph of
$K \seg u$ for some $u \in \nrm (l)$.
By assumption, $G \seg \varphi (u) \in \mathcal{TG(C)}$
for each $u \in \nrm (l)$, and hence
$H \seg v \in \mathcal{TG(C)}$ also holds
for each $v \in \nrm (r_H)$.
On the other hand,
$K \seg v \spt K \seg l$ for each $v \in \safe (r)$.
The induction hypothesis yields
$H \seg v \in \TGnrm$ for each $v \in \safe (r_H)$.
These allow us to conclude
$H \seg r_H \in \TGnrm$.
\end{proof}

For a finite set
$N = \{ m_i \in \mathbb N \mid i \in I \}$,
let $\sum N$ denote the natural
$\sum_{i \in I} m_i$
with the convention
$\sum \emptyset = 0$.
\begin{definition}[Number-theoretic interpretation of term graphs]
\label{d:pint}
\normalfont
Let $G \in \TGnrm$ be a closed term graph over a finite signature
$\FS = \CS \cup \DS$,
$f = \lab_G (\rootnode_G)$, and
$\sp$ be a precedence on $\FS$.
Then, given a positive natural $\ell$, we define an interpretation 
$\pint{\ell} : \TGnrm \rightarrow \mathbb N$
by
\begin{equation*}
\textstyle
\pint{\ell} (G) = \sum
\left\{
  (1+ \ell)^{2 \cdot \rk (f)} \cdot 
  \left( 1 + \sum_{u \in \nrm (v)} \depth (G \seg u) \right)
  \mid 
  v \in \safepath_G \text{ and }
  G \seg v \not\in \mathcal{TG(C)} 
\right\}.
\end{equation*}
\end{definition}
By definition, $\pint{\ell} (G) = 0$ if $G \in \TG{C}$.
We write $\pj{\ell} (G \seg v)$ to abbreviate the component
$(1+ \ell)^{2 \cdot \rk (f)} \cdot 
 \left( 1 + \sum_{u \in \nrm (v)} \depth (G \seg u) \right)$.

\begin{lemma}[Main lemma]
\label{l:pint}
Let $(G, l, r)$ be a constructor rewrite rule such that
$G \seg r \spt G \seg l$ holds for the relation $\spt$
induced by a precedence $\sp$ on a finite signature $\FS$.
Also let
$G_{\mL}, G_{\mR} \in \TGnrm$ respectively denote closed instances of
$G \seg l$ and $G \seg r$.
If $|G \seg r| \leq \ell$, then
$\pint{\ell} (G_{\mR}) < \pint{\ell} (G_{\mL})$
holds for the interpretation $\pint{\ell}$ induced by the precedence $\sp$.
\end{lemma}

\begin{proof}
We estimate an upper bound for
$\pint{\ell} (G_{\mR}) = \sum 
 \{ \pj{\ell} (G \seg v) \mid v \in \safepath_{G_{\mR}}
    \text{ and } G \seg v \not\in \TG{C}
 \}$
dividing the domain
$\{ v \in V_{G_{\mR}} \mid \safepath_{G_{\mR}} 
    \text{ and } G \seg v \not\in \TG{C}
 \}$
into two parts.
Let
$V_{\m{l}} \subseteq V_{G_{\mL}}$
denote the set of labeled nodes that already occur in $G \seg l$.
More precisely, if 
$G_{\mL}$ is the result of instantiation by a homomorphism $\varphi$
from $G \seg l$ to an underlying term graph,
$V_{\m{l}} =
 \{ v \in V_{G_{\mL}} \mid
    \exists u \in V_{G \Seg l} \cap \dom{\lab_G} \text{ s.t. }
    v = \varphi (u)
 \}$.
In other words,
$V_{G_{\mL}} \setminus V_{\m{l}}$
is the set of nodes that are newly added by the instantiation.
Let $V_{\m{r}}$ denote the corresponding subset of
$V_{G_{\mR}}$.
Since every undefined node in $G \seg r$ occurs in
$G \seg l$ as a general assumption, 
$V_{G_{\mR}} \setminus V_{\m{r}} \subseteq
 V_{G_{\mL}} \setminus V_{\m{l}}$
holds.
Write $f$ to denote $\lab_{G} (l)$.

\begin{claim}
\label{c:pj}
$
\pj{\ell} (G_{\mR} \seg v) \leq
(1+ \ell)^{2 \cdot \rk (f) -1} \cdot 
\left( 1 + \sum_{u \in \nrm (\rootnode_{G_{\mL}})} \depth (G_{\mL} \seg u) 
\right)
$
holds for any
$v \in \safepath_{G_{\mR}} \cap V_{\m{r}}$.
\end{claim}

Write $g$ to denote $\lab_{G_{\mR}} (v)$, which is defined by definition
 of $V_{\m{r}}$.
By the assumption $G \seg r \spt G \seg l$,
$g \sp f$ for the given precedence $<$, and hence
$\rk (g) < \rk (f)$ holds.
By Definition \ref{d:ptas}, for each
$v' \in \nrm (v)$, $G_{\mR} \seg v'$ is a sub-term graph of 
$G_{\mL} \seg u$ for some $u \in \nrm (\rootnode_{G_{\mL}})$.
Hence, for each $v' \in \nrm (v)$,
$\depth (G_{\mR} \seg v') \leq
 \sum_{u \in \nrm (\rootnode_{G_{\mL}})} \depth (G_{\mL} \seg u)$,
i.e.,
$
\textstyle
1+ \sum_{u \in \nrm (v)} \depth (G_{\mR} \seg u)
\leq
\textstyle
1+ \arity (g) \cdot
\left( \sum_{u \in \nrm (\rootnode_{G_{\mL}})} \depth (G_{\mL} \seg u)
\right)
\leq
\textstyle
(1+ \ell) \cdot
\left(1+ \sum_{u \in \nrm (\rootnode_{G_{\mL}})} \depth (G_{\mL} \seg u)
\right)
$.
Letting
$v \in \safepath_{G_{\mR}} \cap V_{\m{r}}$,
this allow us to reason as follows.
\begin{eqnarray*}
\pj{\ell} (G_{\mR} \seg v) &\leq&
\textstyle
(1+ \ell)^{2 \cdot \rk (f) -2} \cdot 
\left( 1 + \sum_{u \in \nrm (v)} \depth (G_{\mR} \seg u)
\right)
\quad (\text{since } \rk (g) \leq \rk (f) -1) \\
&\leq&
\textstyle
(1+ \ell)^{2 \cdot \rk (f) -1} \cdot 
\left( 1 + \sum_{u \in \nrm (\rootnode_{G_{\mL}})}
       \depth (G_{\mL} \seg u)
\right)
\end{eqnarray*}
Since
$|\safepath_{G_{\mR}} \cap V_{\m{r}}| \leq |G \seg r| \leq \ell$,
Claim \ref{c:pj} allows us to reason as follows.
\begin{eqnarray}
&&
\textstyle
\sum \left\{ \pj{\ell} (G_{\mR} \seg v) \mid
             v \in \safepath_{G_{\mR}} \cap V_{\m{r}}
             \text{ and } G_{\mR} \seg v \not\in \TG{C}
     \right\}
\nonumber \\
&\leq&
\textstyle
\ell \cdot
(1+ \ell)^{2 \cdot \rk (f) -1} \cdot 
\left( 1 + \sum_{u \in \nrm (\rootnode_{G_{\mL}})}
       \depth (G_{\mL} \seg u)
\right)
\qquad (\text{by Claim \ref{c:pj}})
\nonumber \\
&<&
\textstyle
(1+ \ell)^{2 \cdot \rk (f)} \cdot 
\left( 1 + \sum_{u \in \nrm (\rootnode_{G_{\mL}})}
       \depth (G_{\mL} \seg u)
\right)
= \pj{\ell} (G_{\mL})
\label{e:V_r}
\end{eqnarray}

\begin{claim}
\label{c:V_l}
If
$v \in \safepath_{G_{\mR}} \setminus V_{\m{r}}$
and
$G_{\mR} \seg v \not\in \TG{C}$,    
then
$v \in \safepath_{G_{\mL}} \setminus V_{\m{l}}$
and
$G_{\mL} \seg v \not\in \TG{C}$.
\end{claim}

Suppose
$v \in \safepath_{G_{\mR}} \setminus V_{\m{r}}$
and,
$G_{\mR} \seg v \not\in \TG{C}$.
Then
$G_{\mL} \seg v \not\in \TG{C}$
since
$G_{\mR} \seg v \subg G_{\mL} \seg v$
holds as mentioned above.
We show 
$v \in \safepath_{G_{\mL}} \setminus V_{\m{l}}$
by contradiction.
So assume
$v \not\in \safepath_{G_{\mL}} \setminus V_{\m{l}}$.
Then, since $G_{\mL} \in \TGnrm$,
$G_{\mL} \seg v \in \TG{C}$
holds as observed in the proof of 
Lemma \ref{l:safepath}.\ref{l:safepath:1}.
Claim \ref{c:V_l} allows us to reason as follows.
\begin{eqnarray}
&&
\textstyle
\sum \left\{ \pj{\ell} (G_{\mR} \seg v) \mid
             v \in \safepath_{G_{\mR}} \setminus V_{\m{r}}
             \text{ and } G_{\mR} \seg v \not\in \TG{C}
     \right\}
\nonumber \\
&\leq&
\textstyle
\sum \left\{ \pj{\ell} (G_{\mL} \seg v) \mid
             v \in \safepath_{G_{\mL}} \setminus V_{\m{l}}
             \text{ and } G_{\mL} \seg v \not\in \TG{C}
     \right\}
\qquad (\text{by Claim \ref{c:V_l}})
\nonumber \\
&\leq&
\textstyle
\sum \left\{ \pj{\ell} (G_{\mL} \seg v) \mid
             v \in \safepath_{G_{\mL}} \setminus \{ \rootnode_{G_{\mL}} \}
             \text{ and } G_{\mL} \seg v \not\in \TG{C}
     \right\}
\label{e:V_l}
\end{eqnarray}
Combining the inequalities \eqref{e:V_r} and \eqref{e:V_l},
we conclude
$\pint{\ell} (G_{\mR}) < \pint{\ell} (G_{\mL})$.
\end{proof}

The next lemma states that the {\em normal part} of a starting basic
term graph does not change under precedence termination with argument
separation.
 
\begin{lemma}
\label{l:basic}
Let $\GS$ be a constructor GRS over a signature $\FS$ that is
 precedence-terminating with argument separation 
and
$G_0 \in \mathcal{TG(F)}$
a closed basic term graph.
If $G_0 \rewast G$, then
$G \seg u \ssubnrm G_0$
holds for any nodes
$v \in \safepath_G$ and 
$u \in \nrm (v)$.
\end{lemma}

\begin{proof}
Suppose $G_0 \rewm{n} G$.
We show the assertion by induction on $n \geq 0$.
In the base case $n=0$, $G = G_0$ and 
$\nrm (v) = \emptyset$ for any 
$v \in \safepath_G \setminus \{ \rootnode_G \}$
since $G_0$ is basic.
Hence the assertion trivially holds.

For the induction step, suppose that
$G_0 \rewm{n} G$ holds and that
$G \rew H$ is induced by a redex 
$(R, \varphi)$ in $H$ for a rewrite rule 
$R = (K, l, r) \in \GS$ and a homomorphism
$\varphi : K \seg l \rightarrow G$.
Then 
$G, H \in \TGnrm$
by Lemma \ref{l:safepath}.\ref{l:safepath:2}.
First let us consider the case
$\varphi (l) = \rootnode_G$.
By induction hypothesis,
it suffices to show that for any nodes
$v_H \in \safepath_H$ and
$u_H \in \nrm (v_H)$
there exists a node
$v_G \in \safepath_G$ such that
$H \seg u_H \ssubnrm G \seg v_G$
holds.
Let $v_H \in \safepath_H$ and
$u_H \in \nrm (v_H)$.

{\sc Case.}
$K \seg r \eqspt K \seg u$ for some successor node $u$ of $l$:
Since $\GS$ is a constructor GRS,
$K \seg r \subg K \seg u$, and hence
$H \subg G$ holds.
If $v_H \in \safepath_G$, then we can let
$v_G = v_H$.
If $v_H \not\in \safepath_G$, then 
any path from $\rootnode_G$ to $v_H$ passes an normal argument position of a node
$v_G \in \safepath_G$.
This means
$H \seg v_H \ssubnrm G \seg v_G$, and thus
$H \seg u_G \ssubnrm G \seg v_G$.

{\sc Case.} Otherwise: 
If
$v \in V_G \setminus \{ \varphi (u) \mid u \in V_{G \Seg l }\}$,
then,
as in the previous case, one can find a node
$v_G \in \safepath_G$ such that
$H \seg v_H \ssubnrm G \seg v_G$.
Thus we assume that $v_H$ is mapped from $V_{K \Seg r}$ by $\varphi$.
Then $K \seg r \spt K \seg l$ yields
$H \seg v_H \spt G$.
By the definition of the relation $\spt$,
$H \seg u_H \ssubnrm G \seg \varphi (l)$
holds.
Since $\varphi (l) \in \safepath_G$ by Lemma
\ref{l:safepath}.\ref{l:safepath:1},
we can let $v_G = \varphi (l)$.

Now consider the case
$\varphi (l) \neq \rootnode_H$.
Let
$r_H \in H$ denote the node corresponding to $r \in K$.
Let us consider the subcase $v_H \in V_{H \Seg r_H}$.
In this subcase,
since
$v_H \in \safepath_H (r_H)$,
as in the case
$\varphi (l) = \rootnode_G$,
there exists a node $v_G \in \safepath_G$ such that
$H \seg u_H \ssubnrm G \Seg v_G$
holds.
Since $\varphi (l) \in \safepath_G$ by Lemma
\ref{l:safepath}.\ref{l:safepath:1},
the induction hypothesis yields
$H \seg u_G \ssubnrm G_0$.
Consider the subcase $v_H \not\in V_{H \Seg r_H}$.
In this subcase, $v_H \in \safepath_G$.
As in the previous subcase,
$V_{H \Seg u_H} \cap V_{H \Seg r_H}
 \subseteq
 V_{G \Seg u_H} \cap V_{G \Seg \varphi (l)}$
holds.
On the other side
$V_{H \Seg u_H} \setminus V_{H \Seg r_H} \subseteq
 V_{G \Seg u_H} \setminus V_{G \Seg \varphi (l)}$
holds.
Combining the two inclusions, we
 reason as
\begin{eqnarray*}
\textstyle
 V_{H \Seg u_H}
 =
\textstyle
 \left( V_{H \Seg u_H} \cap V_{H \Seg r_H}
 \right)
 \cup
 \left( V_{H \Seg u_H} \setminus V_{H \Seg r_H}
 \right)
 \subseteq
\textstyle
 \left( V_{G \Seg u_H} \cap V_{G \Seg \varphi (l)}
 \right)
 \cup
 \left( V_{G \Seg u_H} \setminus V_{G \Seg \varphi (l)}
  \right)
 \subseteq
\textstyle
 V_{G \Seg u_H}.
\end{eqnarray*}
This implies
$H \seg u_H \subg G \seg u_H$.
Since $v_H \in \safepath_G$ and
$u_H \in \nrm (v_H)$,
the induction hypothesis yields
$G \seg u_H \ssubnrm G_0$,
and thus
$H \Seg u_H \ssubnrm G_0$.
\end{proof}

To express that a term graph $G$ is maximally shared
{\em with respect to normal argument positions}
of the root $\rootnode_G$,
we define a term graph
$G \cap \nrm$
consisting only of sub-term graphs connecting to 
normal argument positions of $\rootnode_G$.
If $G$ represents a term
$f(\sn{t_1, \dots, t_k}{t_{k+1}, \dots, t_{k+l}})$,
then $G \cap \nrm$ represents the term
$f(\sn{t_1, \dots, t_k}{x_1, \dots, x_l})$
with $l$ fresh variables $x_1, \dots, x_l$.
\begin{definition}
Let $G \in \mathcal{TG(F)}$ be a term graph with
$\att_G (\rootnode_G) =
 \sn{v_1, \dots, v_k}{v_{k+1}, \dots, v_{k+l}}$.
If $\lab_G (v)$ is not defined for any $v \in \safe (\rootnode_G)$,
then $G \cap \nrm$ simply denotes $G$.
Otherwise, $G \cap \nrm$ is defined from $l$ distinct nodes
$u_1, \dots, u_l$ not contained in $V_G$ as follows.

\noindent
$\begin{array}{rcl}
 V_{G \cap \nrm} & = & \textstyle
                    \{ \rootnode_G \} \cup
                   \left( \bigcup_{v \in \nrm (\rootnode_G)} V_{G \Seg v}
                   \right)
                   \cup \{ u_1, \dots, u_l \} \\
 E_{G \cap \nrm} & = & \textstyle
                   \left\{ (u, v) \in E_{G} \mid
                       u, v \in \{ \rootnode_G \} \cup
                   \left( \bigcup_{v \in \nrm (\rootnode_G)} V_{G \Seg v}
                   \right)
                   \right\} \cup
                   \{ (\rootnode_G, u_j) \mid j = 1, \dots, l \} \\
 \lab_{G \cap \nrm} (v) & = &
 \left\{ \begin{array}{ll} 
         \lab_G (v) & \text{if } v \in V_G, \\
         \text{not defined} & \text{otherwise.}
         \end{array}
 \right. \\
 \rootnode_{G \cap \nrm} & = & \rootnode_G
 \end{array}
$

\noindent
By definition,
$\att_{G \cap \nrm} (\rootnode_{G \cap \nrm}) =
 \sn{v_1, \dots, v_k}{u_1, \dots, u_l}$.
A choice of nodes $u_1, \dots, u_l$ is not important and hence will be
 always omitted in later discussions.
\end{definition}

Since an underlying signature 
$\FS = \CS \cup \DS$
is finite, for any (infinite) constructor GRS $\GS$ over $\FS$,
the defined symbols $\DS$ can be partitioned into two sets
$\Dinf$ and $\Dfin$ so that every symbol
$f \in \Dinf$ is defined by an infinite number of rules
whereas every symbol
$f \in \Dfin$ is defined by a finite number of rules.
Accordingly, we define a partition of every constructor GRS $\GS$ into two
sets $\Ginf$ and $\Gfin$
by
$\Ginf = \{ (G, l, r) \in \GS \mid \lab_G (l) \in \Dinf \}$
and
$\Gfin = \{ (G, l, r) \in \GS \mid \lab_G (l) \in \Dfin \}$.
\begin{theorem}
\label{t:context}
Let $\GS$ be a constructor GRS over a finite signature $\FS$ that is
 precedence terminating with argument separation and let
$\max \left( \{ \arity (f) \mid f \in \FS \} \cup
             \{ |K \seg r| \mid \exists l \ (K, l, r) \in \Gfin \}
      \right) \leq d$.
Suppose that, for any rule $(K, l, r) \in \Ginf$, 
{\rm (i)} $(K \seg l) \cap \nrm$ is maximally shared,
{\rm (ii)} $K \seg v$ is closed for every
      $v \in \nrm (l) \cap \dom{\lab_K}$,
{\rm (iii)} $|\{ v \in \nrm (l) \mid v \not\in \dom{\lab_K} \}
              \cup
              (\bigcup_{v \in \safe (l)} V_{K \Seg v})
             | \leq d$, and
{\rm (iv)} $|K \seg r| \leq |K \seg l| + 
       |\bigcup_{v \in \nrm (r)} V_{K \Seg v}|$.
Then, for any closed basic term graph
$G_0 \in \mathcal{TG(F)}$,
if $G_0 \rewast G$, $G \rew H$ and
$2 \cdot
 ( |\bigcup_{v \in \nrm (\rootnode_{G_0})} V_{G_0 \Seg v}| +d )
 \leq \ell
$,
then 
$\pint{\ell} (H) < \pint{\ell} (G)$
holds.
\end{theorem}

We write $(V_G)_{\nrm}$ to abbreviate the set
$\bigcup_{u \in \nrm (\rootnode_G)} V_{G \Seg u}$.
The conditions (i) and (ii) ensure that one step of rewriting can only
reduce a constant number of nodes in
$\Vnrm{K \Seg l}$ by sharing.
The condition (iii) ensures the same for nodes in
$V_{K \Seg l} \setminus \Vnrm{K \Seg l}$.
Since the condition (iv) implies 
$|K \seg r| \leq 2 \cdot |K \seg l|$,
the condition expresses that not only $K \seg l$ but
$K \seg r$ is also suitably shared.

\begin{proof} 
Given a closed basic term graph
$G_0 \in \mathcal{TG(F)}$,
suppose that $G_0 \rewast G$ and that
$G \rew H$ is induced by a redex $(R, \varphi)$
in $G$ for a rule $R = (K, l, r)$ and a homomorphism
$\varphi: K \seg l \rightarrow G$.
Then $G, H \in \mathcal{TG}_{\nrm} (\FS)$
holds by Lemma \ref{l:safepath}.\ref{l:safepath:2}.
Let
$2 \cdot \left( |\Vnrm{G_0}| +d \right) \leq \ell$.
We show that
$|K \seg r| \leq 2 \cdot
 \left( |\Vnrm{G \Seg \varphi (l)}| + d
 \right)
$
holds.
In case $(K, l, r) \in \Gfin$,
$|K \seg r| \leq d$ holds by assumption.
In case $(K, l, r) \in \Ginf$
we deduce the following two inequalities.
\begin{eqnarray}
|K \seg l| \leq |\Vnrm{G \Seg \varphi (l)}| + d
\label{e:context:1} \\
|\Vnrm{K \Seg r}| \leq |\Vnrm{G \Seg \varphi (l)}| +d
\label{e:context:2}
\end{eqnarray}
The homomorphism
$\varphi$ is injective over
$\Vnrm{K \Seg l} \cap \dom{\lab_K}$
by maximal sharing of $(K \seg l) \cap \nrm$.
Hence 
$|K \seg l| \leq |\Vnrm{G \Seg \varphi (l)}| + d$
holds by the assumptions (ii) and (iii).

We deduce the inequality \eqref{e:context:2} by case analysis.
In case that
$K \seg r \eqspt K \seg u$ for some successor node $u$ of $l$,
$K \seg r \in \TG{C}$, and hence
$|\Vnrm{K \Seg r}| = |\emptyset| =0$ 
since constructors only have safe argument positions.
Otherwise,
for every $v \in \nrm (r)$, $K \seg v$ is a sub-term graph of
$K \seg u$ for some $u \in \nrm (l)$.
Thus
$|\Vnrm{K \Seg r}| \leq
 |\Vnrm{K \Seg l}|$
holds, and hence the inequality \eqref{e:context:2} follows from
$|\Vnrm{K \Seg l}| \leq
 |\Vnrm{G \Seg \varphi (l)}| +d
$.

Combining the assumption (iv) with the inequalities
\eqref{e:context:1} and \eqref{e:context:2} yields
$|K \seg r| \leq 2 \cdot
 \left( |\Vnrm{G \Seg \varphi (l)}| + d
 \right)
$.
On the other hand, since 
$\varphi (l) \in \safepath_G$ by Lemma 
\ref{l:safepath}.\ref{l:safepath:1},
Lemma \ref{l:basic} yields
$|\Vnrm{G \Seg \varphi (l)}| \leq 
 |\Vnrm{G_0}|$.
Therefore
$|K \seg r| \leq 2 \cdot
 \left( |\Vnrm{G_0}| + d
 \right)
 \leq \ell$
holds.
Now, letting $r_H \in V_H$ denote the node corresponding to $r \in V_K$,
we deduce
$\pint{\ell} (H) < \pint{\ell} (G)$
as follows.
\begin{eqnarray*}
&&
\pint{\ell} (H) \\
&=&
\textstyle
\pint{\ell} (H \seg r_H)
+
\sum \left\{ \pj{\ell} (H \seg v) \mid
             v \in \safepath_H \setminus \safepath_H (r_H) \text{ and }
             H \seg v \not\in \TG{C}
     \right\}
\\
&<&
\textstyle
\pint{\ell} (G \seg \varphi (l))
+
\sum \left\{ \pj{\ell} (H \seg v) \mid
             v \in \safepath_H \setminus \safepath_H (r_H) \text{ and }
             H \seg v \not\in \TG{C}
     \right\}
\quad (\text{by Lemma \ref{l:pint}})
\\
&\leq&
\textstyle
\pint{\ell} (G \seg \varphi (l))
+
\sum \left\{ \pj{\ell} (G \seg v) \mid
             v \in \safepath_G \setminus \safepath_G (\varphi (l))
             \text{ and }
             H \seg v \not\in \TG{C}
     \right\}
= \pint{\ell} (G)
\end{eqnarray*}
The second inequality follows from
$\safepath_H \setminus \safepath_H (r_H) \subseteq
 \safepath_G \setminus \safepath_G (\varphi (l))$.
\end{proof}

\begin{lemma}
\label{l:size}
Let $\GS$ be a constructor GRS over a finite signature $\FS$ that is
 precedence-terminating with argument separation and let
$\max \left( \{ \arity (f) \mid f \in \FS \} \cup
             \{ |G \seg r| \mid \exists l \ (G, l, r) \in \Gfin \}
      \right) \leq d$.
Suppose the assumptions {\rm (i)--(iv)} in Theorem \ref{t:context}
are fulfilled.
Then, for any closed basic term graph $G_0 \in \mathcal{TG(F)}$,
if
$G_0 \rewm{n} G$, then
$|G| \leq |G_0| +  n \cdot 
 \left( |\bigcup_{v \in \nrm (\rootnode_{G_0})} V_{G_0 \Seg v}| + 2d 
 \right)
$
holds.  
\end{lemma}

\begin{proof} 
By induction on $n$.
In the base case $n=0$, $G = G_0$ and hence the assertion trivially holds.
For the induction step, suppose that
$G_0 \rewm{n} G$ holds and that $G \rew H$
is induced by a redex $(R, \varphi)$ in $G$ for a rule
$R= (K, l, r) \in \GS$ and a homomorphism
$\varphi : K \seg l \rightarrow G$.
In case $R \in \Gfin$,
$|H| \leq |G| +d$
by the choice of the constant $d$.
Suppose $R \in \Ginf$.
As in the proof of Theorem \ref{t:context}, the homomorphism
$\varphi$ is injective over
$\Vnrm{K \Seg l} \cap \dom{\lab_K}$
by maximal sharing of $(K \seg l) \cap \nrm$.
Thus, by the assumptions (ii) and (iii),
at most $d$ nodes in $V_{K \Seg l}$ can be shared by the
homomorphism $\varphi$.
These observations imply
$|H| \leq 
 |G| + |H \seg r_H| - |G \seg \varphi (l)| \leq
 |G| + |K \seg r| + d - |K \seg l|
$
for the node $r_H \in V_H$ corresponding to $r \in V_K$.
Therefore
$|H| \leq |G| + |\Vnrm{K \Seg l}| + d$
holds by the assumption {\rm (iv)}
$|K \seg r| \leq |K \seg l| + |\Vnrm{K \Seg l}|$.
For the same reason as above
$|\Vnrm{K \Seg l}| \leq
 |\Vnrm{G \Seg \varphi (l)}| +d$
holds, and thus
$|H| \leq |G| + |\Vnrm{G \Seg \varphi (l)}| + 2d$
holds.
On the other hand, since
$\varphi (l) \in \safepath_G$ by Lemma
\ref{l:safepath:1}.\ref{l:safepath},
$|\Vnrm{G \Seg \varphi (l)}| \leq
 |\Vnrm{G_0 \Seg \rootnode_{G_0}}|$
holds by Lemma \ref{l:basic}.
Therefore
$|H| \leq |G| + |\Vnrm{G \Seg \rootnode_{G_0}}| + 2d$
holds.
Combining this inequality with the induction hypothesis
$|G| \leq |G_0| + n \cdot 
 \left( |\Vnrm{G_0 \Seg \rootnode_{G_0}}| +2d 
 \right)
$
allows us to conclude
$|H| \leq |G_0| + 
 (n+1) \cdot 
 \left( |\Vnrm{G_0 \Seg \rootnode_{G_0}}| +2d
 \right)
$.
\end{proof}

\begin{corollary}
\label{c:spt}
Suppose that $\GS$ is a constructor GRS over a finite signature $\FS$
 precedence-terminating with argument separation that enjoys
the assumptions {\rm (i)--(iv)} in Theorem \ref{t:context}.
Then there exists a polynomial
$p: \mathbb N \rightarrow \mathbb N$
such that, for any closed basic term graph
${G_0} \in \mathcal{TG(F)}$
and for any term graph $G \in \mathcal{TG(F)}$,
if $G_0 \rewm{m} G$, then the following two conditions hold.
\begin{enumerate}
\item $m \leq p \left( |\bigcup_{v \in \nrm (\rootnode_{G_0})} V_{{G_0} \Seg v}|
      \right)$.
\label{c:spt:1}
\item $|G| \leq p \left( |\bigcup_{v \in \nrm (\rootnode_{G_0})} V_{{G_0} \Seg v}| \right) +
       |V_{G_0} \setminus \bigcup_{v \in \nrm (\rootnode_{G_0})} V_{{G_0} \Seg v}|$.
\label{c:spt:2}
\end{enumerate}
\end{corollary}

\begin{proof}
We only show the existence of a witnessing polynomial
$q: \mathbb N \rightarrow \mathbb N$ for Property \ref{c:spt:1}.
The construction of a polynomial $p$ witnessing both Property
\ref{c:spt:1} and \ref{c:spt:2} will be clear from the polynomial $q$
 and Lemma \ref{l:size}.
Given a GRS $\GS$ over a finite signature $\FS$,
let
$\max \{ \arity (f) \mid f \in \FS \} \leq d$.
In addition, given a closed basic term graph
${G_0} \in \mathcal{TG(F)}$,
let
$2 \cdot
 \left( |\Vnrm{{G_0} \Seg \rootnode_{G_0}}| + d
 \right)
 \leq \ell$.
Suppose that ${G_0} \rewm{m} G$ holds for some term graph
$G \in \mathcal{TG(F)}$.
By Theorem \ref{t:context},
$m$ can be bounded by $\pint{\ell} ({G_0})$.
Since
${G_0} \seg v \in \mathcal{TG(C)}$
for any node
$v \in \safepath_{G_0} \setminus \{ \rootnode_{G_0} \}$,
$\pint{\ell} ({G_0}) = \pj{\ell} ({G_0})$
holds.
Now let $q$ denote a polynomial such that
$(2 x +2d +1)^{2 \cdot \max \{ \rk (f) \mid f \in \FS \}} \cdot (1 + d x) \leq
  q (x)$.
Since
$\sum_{v \in \nrm (\rootnode_{G_0})} \depth ({G_0} \seg v) \leq 
 d \cdot |\Vnrm{{G_0} \Seg \rootnode_{G_0}}|$,
the inequality
$2 \cdot
   \left( |\Vnrm{{G_0} \Seg \rootnode_{G_0}}| + d 
   \right) \leq \ell$
allows us to conclude
$m \leq \pint{\ell} ({G_0}) \leq 
q \left( |\Vnrm{{G_0} \Seg \rootnode_{G_0}}|
  \right)$.
\end{proof}

\begin{remark}
\label{r:extension}
The assumption (iv) in Theorem \ref{t:context} can be
 relaxed as
$|K \seg r| \leq |K \seg l| + p \left( \Vnrm{K \Seg l} \right)$
for some polynomial $p$ if
$\ell$ is sufficiently large so that a certain polynomial in
$|\Vnrm{G_0}|$ determined by $p$ can be bounded by $\ell$.
Since such a relaxed form of the condition (iv) likely holds under a suitable
 term rewriting adoption of unfolding graph rewrite rules,
it turns out that just unfolding a recursion schema seems not crucial to deduce
 the polynomial complexity.
But, more importantly, as implied from the assumption (iii), the number of variables
 occurring in the right-hand side of every rule can be constantly
 bounded, which clearly fails in any reasonable term rewriting
 formulation of unfolding rewrite rules.
\end{remark}
The next lemma ensures that the assumption (i) in Theorem
\ref{t:context} is not too restrictive.

\begin{lemma}
\label{l:Gnrm}
Let $\GS$ be a constructor GRS over a finite signature $\FS$
precedence-terminating with argument separation.
For any maximally shared, closed basic term graph
$G_0 \in \TG{F}$, if  
$G_0 \rewast G$, then
$(G \seg v) \cap \nrm$ is maximally shared for any
$v \in \safepath_G$.
\end{lemma}
 
\begin{proof}
Let $v \in \safepath_G$ and
$u_0, u_1 \in V_{G \Seg v}$.
Assume 
$\term_G (G \seg u_0) = \term_G (G \seg u_1)$.
By the definition of the term graph
$(G \seg v) \cap \nrm$,
it suffices to consider the case
$u_0, u_1 \in \Vnrm{G \Seg v}$.
In this case, by Lemma \ref{l:basic},
$G \seg u_j \ssubnrm G_0$
holds for each $j = 0, 1$.
This means that
$G \seg u_j = G_0 \seg u_j$
holds for each $j=0, 1$, and thus
$\term_{G_0} (G_0 \seg u_0) =
 \term_{G_0} (G_0 \seg u_1)$
holds by the assumption.
Maximal sharing of $G_0$ implies $u_0 = u_1$.
\end{proof}

As a consequence of Lemma
\ref{l:safepath:1}.\ref{l:safepath}
and Lemma \ref{l:Gnrm},
for any (completely defined) constructor GRS $\GS$ over a finite
signature that is precedence terminating with argument separation,
if there exists a constant $d$ such that
the assumptions (i)--(iv) in Theorem \ref{t:context} hold
for any rule $(K, l, r) \in \Ginf$, then
any rewriting sequence
$G_0 \rew G_1 \rew \cdots$
starting with a maximally shared, closed basic term graph
$G_0$ leads to a constructor term graph in normal form.

\begin{theorem}
\label{l:gsr}
Every general safe recursive function can be computed by a
 constructor GRS that precedence terminating with an argument separation 
fulfilling the conditions {\rm (i)--(iv)} in Theorem \ref{t:context}.
\end{theorem}

\begin{proof}
By induction over the definition of $f$.
In the base case, every initial function can be defined by a single
 constructor rewrite rule $(G, l, r)$ in
one of the following shapes \ref{d:srgrs:1} and \ref{d:srgrs:2}.
\begin{enumerate}
\item $G \seg r = G \seg v$ for some successor node
      $v$ of $l$.
\label{d:srgrs:1}
\item 
\label{d:srgrs:2}
        $V_G$ consists of 
        $2+k+l+n$ elements
        $u$, $v$, $x_1, \dots, x_{k+l}$,
        $w_1, \dots, w_n$ such that
        $l = u$, $r = v$,
  \begin{itemize}
  \item $\{ \lab_G (u), \lab_G (v), \lab_G (w_1), \dots, \lab_G (w_n)
         \} \subseteq \FS$,
  \item $\lab_G (x_j)$ is undefined for all 
        $j \in \{ 1, \dots, k+l \}$,
  \item $\att_G (u) =
         \sn{x_1, \dots, x_k}{x_{k+1}, \dots, x_{k+l}}$,
  \item $\att_G (v) =
         \sn{x_{j_1}, \dots, x_{j_m}}{w_1, \dots, w_n}$
        for some
        $\{ j_1, \dots, j_m \} \subseteq \{ 1, \dots, k \}$, and
  \item $\att_G (w_j) =
         \sn{x_1, \dots, x_k}{x_{k+l}, \dots, x_{k+l}}$
        for all $j \in \{ 1, \dots, n \}$.
  \end{itemize}
\end{enumerate}
The graph rewrite rule (1) below is an instance of Case 
\ref{d:srgrs:2} with $k=2$, $l=1$ and $n=2$,
which expresses the term rewrite rule
$\mf (\sn{x_1, x_2}{x_3}) \rightarrow
 \mh (\sn{x_1}{\mg_1 (\sn{x_1, x_2}{x_3}),
               \mg_2 (\sn{x_1, x_2}{x_3})
              }
     )$.
As in Figure \ref{fig:srugr}, every edge
$\xymatrix{v \ar[r] & u}$ is expressed as
$\xymatrix{v \ar@{-->}[r] & u}$ if
$u \in \safe (v)$ and $\lab_G (v) \in \DS$.
\\
$\mbox{}$ \hfill
$
  \xymatrix{(1) & &
            *+[F]{\mh} \ar@/_/[ddll] \ar@{-->}[dl] \ar@{-->}[d] \\
            *+[o][F-]{\mf} \ar[d] \ar[dr] \ar@{-->}[drr] &
            \mg_1 \ar[dl] \ar[d] \ar@{-->}[dr] &
            \mg_2 \ar[dll] \ar[dl] \ar@{-->}[d] \\
            \bot & \bot & \bot
           }
\qquad
  \xymatrix{(2) &
            *+[o][F-]{\m{o}^{2,2}_j} \ar[dl] \ar[d]
            \ar@{-->}[dr] \ar@{-->}[drr] & &
             *+[F]{\m{c}_j} \\
            \bot & \bot & \bot & \bot
           }
\qquad
  \xymatrix{(3) &
            *+[o][F-]{\mI{2,2}{3}} \ar[dl] \ar[d]
            \ar@{-->}[dr] \ar@{-->}[drr] & & \\
            \bot & \bot & *+[F]{\bot} & \bot
           }
$
\hfill $\mbox{}$ \\
Every instance of {\bf (Constants)} can be defined by a single graph rewrite
 rule as (2) above in a special shape of Case \ref{d:srgrs:2}, and each of
{\bf (Projections)}, {\bf (Predecessors)} and {\bf (Conditional)}
can be defined by a single graph rewrite rule as (3) in the form of 
Case \ref{d:srgrs:1}.
The induction step splits into two cases.
In case that $f$ is defined by {\bf (Safe composition)},
$f$ is defined by a constructor graph rewrite rule in the form of 
Case \ref{d:srgrs:2} together with the
constructor GRSs obtained from induction hypothesis.
In case that $f$ is defined by {\bf (General safe recursion)},
$f$ is defined by an infinite set of constructor safe recursive unfolding graph
 rewrite rules together with the constructor GRSs
 obtained from induction hypothesis.
For instance, suppose that $f$ is defined by
$f(\sn{\epsilon}{z}) = g(\sn{}{z})$
and
$f(\sn{c (\sn{}{x, y})}{z}) =
 h(\sn{x, y}{z, f(\sn{x}{z}), f(\sn{y}{z})}
  )$.
By induction hypothesis, $g$ and $h$ can be respectively computed by
some constructor GRSs $\GS_{\m{g}}$ and $\GS_{\m{h}}$ defining the
 corresponding function symbols
$\mg, \mh \in \DS$.
Let
$\m{e}, \m{c} \in \CS$
respectively correspond to $\epsilon, c$
and $\mf \in \DS$ to $f$
Also let
$\Sigma = \{ \m{e}, \m{c} \}$
and
$\Theta = \{ \mg, \mh \}$
with a bijective correspondence
$\m{e} \mapsto \mg, \m{c} \mapsto \mh$.
Then, for each $m \geq 1$, one can define the $m^{\text{th}}$ set $\GS_m$ of safe recursive unfolding
 graph rewrite rules over $\Sigma \cup \Theta$ defining
$\mf$.
Since $\Sigma \subseteq \CS$, $\GS_m$ is a constructor GRS for every
$m \geq 1$.
Since elements of $\bigcup_{m \geq 1} \GS_m$ express
$\mf (\sn{\m{e}}{z}) \rightarrow \mg (\sn{}{z})$,
$\mf (\sn{\m{c} (\sn{}{\m{e}, \m{e}})}{z}
     ) \rightarrow
 \mh (\sn{\m{e}, \m{e}}{z, \mg (\sn{}{z}), \mg (\sn{}{z})}
     )$,
$\ldots$,
$f$ is computed by the infinite GRS
$\GS_{\mg} \cup \GS_{\mh} \cup 
 \left( \bigcup_{m \geq 1} \GS_m \right)$.

The precedence $\sp$ is defined so that, letting every constructor be
$\sp$-minimal,
 for every rule $(G, l, r)$,
$\lab_G (v) \sp \lab_G (l)$ for any $v \in V_{G \Seg r}$
whenever $\lab_G (v)$ is defined.
Then $\mg \sp \mf$ means that $f$ is defined from $g$ for the
 functions $f, g$ respectively corresponding to $\mf, \mg$.
Hence the well-foundedness of $\sp$ follows from the observation that
the relation ``is defined from'' is well-founded by the definition of
 general safe recursive functions.
Precedence termination of so obtained GRSs is obvious.

Let $\spt$ be the relation induced by the precedence $\sp$.
By definition, the subset $\Ginf$ of $\GS$ consists of safe recursive
 unfolding graph rewrite rules whereas $\Gfin$ contains no unfolding
 graph rewrite rule.
It follows from the definition of safe recursive unfolding graph rewrite
 rules that
$G \seg l \spt G \seg r$
for each $(G, l, r) \in \Ginf$
(See also Corollary \ref{c:sruf}).
Consider a rewrite rule $(G, l, r) \in \Gfin$.
It is obvious that 
$G \seg l \spt G \seg r$
holds if $(G, l, r)$ is an instance of Case \ref{d:srgrs:1}.
Suppose that $V_G$ consists of $2+k+l+n$ elements
$l$, $r$, $x_1, \dots, x_{k+l}$, $w_1, \dots, w_n$
as specified in Case \ref{d:srgrs:2}. 
Let
$v \in V_{G \Seg r} = 
 \{ r, x_1, \dots, x_{k+l}, w_1, \dots, w_n \}$.
Consider the case that $\lab_G (v)$ is not defined, i.e.,
$v \in \{ x_1, \dots, x_{k+l} \}$.
In this case, $v$ is a successor node of $l$.
Namely $G \seg v = G \seg u$ for some successor node $u$ of $l$, and
 hence $G \seg l \spt G \seg v$ holds.
Assume that $\lab_G (v) \in \FS$.
Then $v \in \{ r, w_1, \dots, w_n \}$.
Since
$\att_G (w_j) = \sn{x_1, \dots, x_k}{x_{k+1}, \dots, x_{k+l}}$
for every $j \in \{ 1, \dots, n \}$,
$G \seg l \spt G \seg w_j$
for every $j \in \{ 1, \dots, n \}$.
This yields $G \seg l \spt G \seg v$
since 
 $\att_G (v) = \sn{x_{j_1}, \dots, x_{j_m}}{w_1, \dots, w_n}$
for some
$\{ j_1, \dots, j_m \} \subseteq \{ 1, \dots, k \}$.
The conditions (ii)--(iv) follow from the definition of unfolding
 graph rewrite rules.
Choosing every rewrite rule $(G, l, r) \in \GS$ so that
$(G \seg l) \cap \nrm$ is maximally
 shared allows one to conclude.
\end{proof}
\begin{corollary}
\label{c:main}
For every general safe recursive function $f$, there exist a constructor
 GRS $\GS$ that computes $f$ and a polynomial
$p: \mathbb N \rightarrow \mathbb N$
such that, for any maximally shared, closed basic term graph $G$,
if $G \rewm{m} H$, then
$m \leq p(n)$ and
$|H| \leq p (n) + |G|$
hold, where $n$ denotes the size
$|\bigcup_{v \in \nrm (\rootnode_G)} V_{G \Seg v}|$
(of the union) of the subgraphs connected to
the normal argument positions of $\rootnode_G$ only.
\end{corollary}
The corollary says that every general safe recursive function can be computed by a polynomially bounded constructor GRS.
Since such a witnessing GRS is polytime presentable in particular, Corollary
\ref{c:main} yields an alternative proof of Theorem \ref{t:GRR10}.

%% file: application.tex
\section{Related works and further application}
\label{s:app}

In this section we discuss two related works to see some
potential applicability of the method presented in the previous section
and one more work to see a limit of the computational power of the
method.

In \cite{Marion03} a term rewrite system $\RS_{\m{lcs}}$, which computes
the length of the
{\em longest common subs-sequence} of two strings, is discussed.
The rewrite system $\RS_{\m{lcs}}$ contains instances of
\\
$\mbox{}$ \hfill
$\begin{array}{rcl}
 \mf (\sn{\epsilon, y, z}{w}) & \rightarrow &
 \mg (\sn{y, z}{w})
 \qquad \qquad
 \mf (\sn{x, \epsilon, z}{w}) \quad \rightarrow \quad
 \mg (\sn{x, z}{w})
 \\
 \mf (\sn{\m{c}_i (x), \m{c}_j (y), z}{w}) & \rightarrow &
 \mh_{i, j} (\sn{x, y, z}{w, \mf (\sn{x, \m{c}_j (y), z}{w}),
                             \mf (\sn{\m{c}_i (x), y, z}{w})
                         }
            ),
 \end{array}
$
\hfill $\mbox{}$
\\
i.e., rewrite rules expressing safe recursion with multiple recursion
arguments.
For exactly the same reason as in case of general safe recursion, 
$\RS_{\m{lcs}}$ only admits a polynomial {\em quasi}-interpretation
which says nothing about polynomial runtime complexity.
Due to the restriction to single recursion arguments, it is not possible to represent these rules as 
instances of (safe recursive) unfolding graph rewrite rules.
However, as seen from an instance (1) below
(where the variable $z$ is ignored to ease the presentation), 
$\RS_{\m{lcs}}$ could be represented by an infinite GRS
fulfilling the assumptions (i)--(iv) in Theorem \ref{t:context}.
\\
$\mbox{}$ \hfill
$
  \xymatrix{(1) &
            *+[o][F-]{\mf} \ar[dl] \ar[d] \ar@/^/@{-->}[dd] & & \\
            \m{c}_i \ar@{..>}[d] & \m{c}_j \ar@{..>}[dl] & &
            *+[F]{\mh_{i, j}} \ar@/_/[dlll] \ar[dlll]
            \ar@{-->}[dll] \ar@{-->}[dl] \ar@{-->}[d] \\
            \epsilon & \bot & \mg \ar@{-->}[l] &
            \mg \ar@{-->}@/^/[ll]
           }
\qquad \qquad
  \xymatrix{(2) &
            *+[o][F-]{\mf_i} \ar[dl] \ar[dd] \ar@{-->}[ddr] & & & \\
            \m{c} \ar@{..>}@/_/[d] \ar@{..>}@/^/[d] & & &
            *+[F]{\mh_{i}} \ar@/_/[dlll] \ar[dlll] \ar[dll]  
            \ar@{-->}[dl]
            \ar@{-->}@/_/[d] \ar@{-->}@/^/[d]
            \ar@{-->}[dr] \ar@{-->}@/^/[dr] & \\
            \epsilon & \bot & \bot &
            \mg_0 \ar@/_/[ll] \ar@{-->}[l] & 
            \mg_1 \ar@/_5mm/[lll] \ar@{-->}@/^/[ll]
           }
$
\hfill $\mbox{}$

In a very recent work \cite{ADL15},
Theorem \ref{t:GRR10} is expanded for 
{\em simultaneous} general safe recursion, e.g.,

$\begin{array}{rcl}
 \mf_i (\sn{\epsilon, z}{w}) & \rightarrow & \mg_i (\sn{z}{w})
 \\
 \mf_i (\sn{\m{c} (x, y), z}{w}) & \rightarrow &
 \mh_i (\sn{x, y, z}{w,
            \mf_0 (\sn{x, z}{w}), \mf_0 (\sn{y, z}{w}),
            \mf_1 (\sn{x, z}{w}), \mf_1 (\sn{y, z}{w})
           }
       )
 \end{array}
$
$(i = 0, 1)$. \\
In contrast to the current approach, instead of taking an advantage of
sharing in term graph rewriting, the notion of
{\em cache} is employed in \cite{ADL15} to avoid costly recomputations.
A similar notion, called {\em minimal function graphs}, can be found in
\cite{Marion03}, yielding that the rewrite system $\RS_{\m{lcs}}$ can be
executed in polynomial time.
As mentioned in Remark \ref{r:extension}, the condition (iv) in Theorem
\ref{t:context} can be relaxed as (iv)'
$|K \seg r| \leq |K \seg l| + O (|\Vnrm{K \Seg l}|)$.
Thus, as seen from an instance (2) above,
such the schema of simultaneous recursion could be also represented by an
infinite GRS enjoying the assumptions (i)--(iii) and (iv)'.

As shown in \cite{LM94}, it is known that the polynomial-space
computable functions can be captured with safe recursion (on notation)
{\em with parameter substitutions}.
To see an explicit boundary of the proposed method, consider the term rewrite
system below that expresses an instance of the schema.
\\
$\mbox{}$ \hfill
$\begin{array}{rcl}
 \mf (\sn{\epsilon}{y}) & \rightarrow &
 \mg (\sn{}{y}) \\
 \mf (\sn{\m{c} (x)}{y}) & \rightarrow &
 \mh (\sn{x}{y, \mf (\sn{x}{\m{p} (\sn{x}{y})}),
                \mf (\sn{x}{\m{q} (\sn{x}{y})})
               })
 \end{array}
$
\hfill \ \\
The rules below are the first three instances of unfolding the above rules.

$\begin{array}{rrcl}
 (0) &
 \mf (\sn{\epsilon}{y}) & \rightarrow &
 \mg (\sn{}{y}) \\
 (1) &
 \mf (\sn{\m{c} (\epsilon)}{y}) & \rightarrow &
 \mh (\sn{\epsilon}{y, \mg (\sn{}{\m{p} (\sn{\epsilon}{y})}),
                   \mg (\sn{}{\m{q} (\sn{\epsilon}{y})})
               }) \\
 (2) &
 \mf (\sn{\m{c} (\m{c} (\epsilon))}{y}) & \rightarrow &
 \mh (\m{c} (\epsilon); y,
      \mh (\epsilon; y,
           \mg (; \m{p} (\epsilon; 
                         \m{p} (\m{c} (\epsilon), y)
                        )
               ),
           \mg (; \m{q} (\epsilon; 
                         \m{p} (\m{c} (\epsilon), y)
                        )
               )
          ), \\
 &&& \hspace{1.5cm}
      \mh (\epsilon; y,
           \mg (; \m{p} (\epsilon; 
                         \m{q} (\m{c} (\epsilon), y)
                        )
               ),
           \mg (; \m{q} (\epsilon; 
                         \m{q} (\m{c} (\epsilon), y)
                        )
               )
          )
     )
 \end{array}
$
\\
One will see that $\mg$ occurs $2^i$ times in the
$i^{\text{th}}$ instance $(i)$ and none of the occurrences can be shared
since their arguments are different.
For this reason, even if they are represented as maximally shared GRS
$\GS$,
$2^{|\Vnrm{K \Seg l}|} \leq |K \seg r|$
for every $(K, l, r) \in \Ginf$,
and thus (even a relaxed form of) the condition (iv) fails.